\documentclass{article}
\usepackage{graphicx}
\usepackage{amsmath}
\usepackage{amssymb}
\usepackage{amsthm}
\usepackage{subcaption}
\usepackage{authblk}

\newtheorem{proposition}{Proposition}

\title{Exact learning of quantum noise with tensor networks}

\author[1]{Nicola Pancotti}
\author[1]{Vedika Saravanan}
\author[1]{Krysta Svore}
\affil[1]{NVIDIA Corporation, 2788 San Tomas Expressway, Santa Clara, 95051, CA, USA}

\begin{document}

\maketitle

\begin{abstract}
Accurate noise models are essential for high-performance quantum error correction, yet characterizing the noise of a quantum device typically requires dedicated experiments.
We present a variational framework that learns the noise model directly from quantum-error-correction syndrome and logical-observable data collected during error-corrected memory experiments.
The fault-event probabilities are treated as variational parameters and optimized via gradient descent to minimize the binary cross-entropy between the decoder's predictions and experimental logical-observable outcomes.
We prove that this objective is principled rather than ad hoc: a sufficiently expressive noise ansatz attains the information-theoretic minimum logical error rate.
We instantiate this framework using a tensor-network decoder, which provides exact maximum-likelihood decoding and analytically differentiable gradients with respect to all noise parameters.
Using circuit-level data from Google's Sycamore processor, and starting from an uninformed prior, the optimization recovers noise models whose logical error rates agree to within $2\%$ with those of Google's independently characterized detector error model.
The mean squared error between the learned and reference noise parameters shows a clear overall decrease throughout training, confirming that the method recovers physically meaningful noise structure, not merely parameters that happen to decode well.
We further demonstrate that the optimization can track device drifts in real time via warm-started updates, maintaining near-optimal decoding performance under synthetically evolving noise without the need for re-characterization.
The approach is decoder-agnostic in its formulation and naturally extends to correlated noise models.
\end{abstract}

\section{Introduction}

Fault-tolerant quantum computation relies on quantum error correction (QEC) to protect logical information from physical noise~\cite{dennis2002topological}.
Recent experiments on superconducting processors have demonstrated QEC performance below the surface-code threshold~\cite{acharya2023suppressing, acharya2024quantum}, marking critical milestones toward practical fault tolerance.
For a given quantum error-correcting code, a central operation is \textit{syndrome extraction}: at each round, a set of stabilizer checks is measured, producing classical syndrome data that is fed to a decoder.
The decoder's task is to identify a likely correction that is consistent with the syndrome and preserves the logical information.

The performance of a decoder depends critically on the accuracy of the assumed underlying \textit{noise model}, a probability distribution over fault configurations that encodes the likelihood of each error mechanism.
In the detector error model (DEM) framework~\cite{gidney2021stim}, this information is captured by a set of fault-event priors associated with the edges of a decoding graph.
When these priors faithfully reflect the true hardware-device noise, decoders such as Minimum Weight Perfect Matching~\cite{higgott2023sparse}, Union-Find~\cite{delfosse2021almost}, Belief Propagation~\cite{roffe2020decoding}, and tensor-network methods~\cite{bravyi2014efficient, chubb2021statistical, piveteau2024tensor, bohdanowicz2022quantum} can achieve near-optimal logical error rates.
On real hardware, however, the priors are often only approximately known, and decoder performance degrades as the mismatch between the assumed and true noise model grows.

Characterizing the noise of a quantum hardware device is itself a substantial challenge.
Traditional tomographic protocols provide rigorous descriptions of errors but suffer from exponential resource scaling~\cite{mohseni2008quantum}.
Scalable alternatives such as randomized benchmarking~\cite{wallman2016noise}, Pauli channel estimation~\cite{flammia2020efficient, harper2020efficient, flammia2021pauli}, and sparse Pauli--Lindblad models~\cite{vandenberg2023probabilistic} have been developed, but these ``offline'' methods require dedicated characterization experiments.
In situ benchmarking approaches have recently been extended to fault-tolerant Clifford circuits~\cite{xiao2026situ}, though they still require additional experimental overhead, including dedicated benchmarking runs and extra classical post-processing steps that are separate from routine memory-experiment decoding.
Accurate noise knowledge is also central to quantum error mitigation, as reviewed in Ref.~\cite{cai2023quantum}.

A natural question is whether the syndrome data already collected during an error-correction experiment on a quantum hardware device can itself be used to learn the noise model, eliminating the need for separate characterization.
Early work by Laforest et al.~\cite{laforest2007using} already showed how experimental error-correction data can be used to infer noise-model information.
This viewpoint has been explored in a series of works by Wagner et al.~\cite{wagner2021optimal, wagner2022pauli, wagner2023learning}, who showed that under phenomenological noise models, the logical Pauli channel can be inferred from syndrome statistics alone.
Zheng et al.~\cite{zheng2025efficient} recently extended this framework to circuit-level noise, deriving necessary and sufficient conditions for learnability and providing efficient estimation protocols with provable sample-complexity guarantees.
In parallel, several groups have demonstrated that the DEM prior distribution can be estimated directly from syndrome data~\cite{blumekohout2025estimating, sivak2024optimization, iyer2025enhancing}, and adaptive strategies that re-weight decoding graphs online have shown improved decoder performance under drifting or correlated noise~\cite{nickerson2019analysing, wang2024dgr, spitz2018adaptive}.
Neural-network decoders trained end-to-end on syndrome data~\cite{torlai2017neural, bausch2024learning} represent yet another approach, though they typically bypass the noise model without providing an interpretable description of device errors.

In this work, we take a complementary approach: we treat the noise model parameters as variational degrees of freedom and optimize them to maximize decoding performance directly on experimentally obtained or simulated syndrome data.
Rather than inferring the noise channel from syndrome statistics via analytical inversion, we formulate the problem as a variational optimization in which the loss function is a differentiable surrogate of decoder performance: the binary cross-entropy between model-predicted and experimentally observed logical outcomes on a syndrome dataset.
This formulation is both decoder- and code-agnostic: any decoder that produces logical-observable probabilities can serve as the computational backbone, and our approach applies to any stabilizer code admitting a parity-check representation.

We further prove that this variational principle is information-theoretically optimal: the population minimizer of the binary cross-entropy is the true Bayes-optimal posterior decoder and thus attains the minimum achievable logical error rate.

We instantiate this framework using a circuit-level tensor-network (TN) decoder~\cite{bravyi2014efficient, piveteau2024tensor}, which provides two key advantages for noise learning.
First, TN contractions yield exact or systematically improvable approximations to the maximum-likelihood decoding problem, producing high-quality gradient signals.
Second, the entire pipeline, from syndrome input to logical-observable prediction, is analytically differentiable with respect to the noise model parameters, enabling efficient variational optimization without finite-difference approximations.

The rest of this paper is organized as follows.
In Sec.~\ref{sec:general_framework}, we set up the mathematical framework for decoding in terms of parity-check matrices, noise models, and logical observables.
In Sec.~\ref{sec:NML}, we formulate the noise-learning problem as a variational optimization and derive the relevant gradients.
In Sec.~\ref{sec:TN_decoder}, we specialize to tensor-network decoders and show how the full pipeline, decoding and noise learning, maps onto tensor-network contractions.
We demonstrate the method on realistic circuit-level noise data extracted from Google's Sycamore processor~\cite{acharya2023suppressing, acharya2023suppressing_data}, showing that competitive noise models can be learned from scratch using only syndrome data and then used online to track device drifts.

\section{General framework}\label{sec:general_framework}

Let us formalize the decoding problem starting from its core building blocks: the code, the noise model, and the logical observables.
The problem can be mathematically formulated starting from a \textit{parity-check matrix} $H$, a binary matrix over $GF(2)$ that represents either the quantum error-correcting code or a circuit-level noise scenario through a detector error model.
Given a syndrome vector $s \in GF(2)$, the linear system
\begin{equation}\label{eq:linear_boolean_system}
    He=s \pmod{2}
\end{equation}
forms one of the core building blocks, before the noise model and the logical observables are chosen.
All solutions $e$ of Eq.~\eqref{eq:linear_boolean_system} are errors compatible with the syndrome $s$.
The matrix $H$ can also be seen as the adjacency matrix of a bipartite, undirected graph often called the Tanner graph.
Rows and columns of $H$ form the two subsets of vertices of the bipartite graph.
The non-zero entries of $H$ are the edges connecting the two subsets of vertices.
For example, consider
\begin{equation}\label{eq:example_parity_check}
    H = \begin{bmatrix}
0 & 0 & 0 & 1 & 1 \\
0 & 1 & 1 & 1 & 0 \\
1 & 0 & 1 & 0 & 1
\end{bmatrix},
\end{equation}
with the corresponding graph shown in Fig.~\ref{fig:tn_mapping}a).
It is possible to define a probability distribution over all solutions of Eq.~\eqref{eq:linear_boolean_system}~\cite{mezard2009information}.
Let us define the indicator function
\begin{equation}\label{eq:indicator_function}
\begin{split}
    I_a(x_1, x_2, \ldots, x_{d_a+1}) &= 1 \text{, if } x_1 \oplus x_2 \oplus \cdots \oplus x_{d_a+1} = 0, \\
    & = 0 \text{, otherwise}.
\end{split}
\end{equation}

The (unnormalized) indicator-function representation of the solution set of Eq.~\eqref{eq:linear_boolean_system}, for the special case in Eq.~\eqref{eq:example_parity_check} reads
\begin{equation}
    P_H (e, s) = I_1(e_4, e_5, s_1) \cdot I_2(e_2, e_3, e_4, s_2) \cdot I_3(e_1, e_3, e_5, s_3),
\end{equation}
which can be generalized to arbitrary parity-check matrices as
\begin{equation}\label{eq:P_H}
    P_H (e, s) = \prod_a I_a\!\left(\{e_i\}_{i \in \partial a}, s_a\right).
\end{equation}
Given a syndrome $s'$, it is easy to see that any exact sample $e'$ from $P_H (e, s')$ satisfies $He' = s'$.

The next building block is the noise model.
As we discussed above, the noise model is nothing more than a distribution over all possible error mechanisms $P_N (e)$.
In the probabilistic language introduced in Eq.~\eqref{eq:P_H}, it is easy to combine the noise model with the code as
\begin{equation}\label{eq:code_noise}
    \text{Code + Noise} \rightarrow P_{H, N}(e, s) = P_H (e, s) P_N (e),
\end{equation}
where each error mechanism $e'$ is weighted by $P_N(e')$.
Note that $P_{H, N}(e, s)$ already contains all the necessary ingredients to solve decoding and inference tasks.
After fixing the syndrome $s$, sampling from it corresponds to searching for the most-likely error.
There are numerous decoders that use a simple ansatz for the noise model and attempt to sample Eq.~\eqref{eq:code_noise} approximately or exactly.

In order to introduce an \textit{exact} logical error rate, we need one last building block, the logical observables.
Any logical observable can be modeled as an additional row in the parity-check matrix.
It is straightforward to define an additional factor
\begin{equation}\label{eq:P_L}
    P_L(e,\ell)=\prod_{b=1}^{k} I_b\!\left(\{e_i\}_{i\in\partial b},\ell_b\right),
\end{equation}
where the index $b$ runs over the logical observables $\ell_b$.
As previously, we can combine $P_L (e, \ell)$ with the rest of the problem as
\begin{equation}\label{eq:full_distribution}
    \text{Code + Noise + Logicals} \rightarrow P(e, s, \ell) = P_H (e, s) P_N (e) P_L (e, \ell).
\end{equation}

In what follows, we restrict ourselves to a single logical observable $\ell_1$.
With that assumption, we ask for the probability that a given syndrome $s$ flips $\ell_1$.
This quantity can be expressed exactly, and it boils down to computing the marginal of $\ell_1$, that is,
\begin{equation}\label{eq:l_mariginal}
    p(s, \ell_1) = \sum_e P (e, s, \ell_1).
\end{equation}
In the maximum likelihood decoding scenario, we say that the syndrome $s$ has flipped the logical $\ell_1$ if $p(s, 1) > p(s, 0)$.
Notice that in Eq.~\eqref{eq:l_mariginal} we sum over all error mechanisms weighted by the noise model; we do not select a single sample as in the most-likely error case above.
Most-likely error decoders approximate maximum likelihood decoders with the assumption that $p(s, \ell_1) \sim P_L (e_M (s), \ell_1)$ where $e_M (s) = \text{argmax}_e P_{H, N}(e, s)$ is an error bit string with the highest probability of occurrence.

A central quantity that is often used as a proxy to judge a decoder's performance is the Logical Error Rate (LER)~\cite{dennis2002topological}.
Given $\mathcal D=\{(s_n,y_n)\}_{n=1}^{|\mathcal D|}$, let $\widehat y_n$ be one if $p(s_n,1)>p(s_n,0)$ and zero otherwise.
We define the LER as
\begin{equation}\label{eq:ler}
    \mathrm{LER} = \frac{1}{|\mathcal D|} \sum_{n=1}^{|\mathcal D|} \left(\widehat y_n-y_n\right)^2,
\end{equation}
which is the fraction of disagreements with the observed outcome.
If the $p(s, \ell_1)$'s are computed exactly, the LER reaches its minimum theoretical value, which is not necessarily zero due to statistical fluctuations.
On the other hand, when the $p(s, \ell_1)$'s are not exact, the LER typically exhibits larger values.

\section{Noise model learning}\label{sec:NML}

As we discussed above, the noise model plays a central role in the decoding problem.
For theoretical studies, one would typically assume a tractable noise model, which is modeled as an explicit functional form of $P_N(e)$.
Common choices are uncorrelated noise models $P_N (e) = \prod_i P_N (e_i)$, which model error mechanisms as independent.
That choice lends itself to efficient numerical treatments, opening the door to various algorithmic decoders such as Minimum Weight Perfect Matching and Belief Propagation.
In any realistic scenario, the extraction of an accurate noise model of the hardware device is often a daunting problem.
Error mechanisms are typically subtle and correlated, rendering their statistical impact on error correction protocols opaque.

In this section, we propose a general framework for optimizing noise models for error-correction applications.
The idea is simple: given a code and, for pedagogical purposes, one logical observable (the framework generalizes to more), we seek noise models that minimize relevant QEC figures of merit, such as the logical error rate.
Let
\begin{equation}\label{eq:log_flip_prob}
    x_s = \frac{p(s, 1)}{p(s, 1) + p(s, 0)}
\end{equation}
be the prediction of the decoder that the syndrome $s$ flipped the logical observable.
We want $x_s$ to be as close as possible to $y_s$, the experimental outcome.
We assume that the $y_s$'s are already normalized.
Notice that all quantities can be provided as soft probabilities, including the syndromes and experimental logical flips.
Furthermore, from Eq.~\eqref{eq:l_mariginal} and Eq.~\eqref{eq:full_distribution}, $x_s$ contains the full dependence on the noise model.
We can choose a noise model parameterization that allows us to carry out this task.
One of the simplest choices is $P_N(e; \theta) = \prod_i P_i(e_i; \theta_i)$ which, in the case of binary errors $e_i \in \{0, 1\}$, reduces to one scalar parameter $\theta_i$ per error $e_i$: $P_i(e_i; \theta_i) = \theta_i^{e_i} (1 - \theta_i)^{1-e_i}$.
In the remainder of this section we do not restrict $P_N(e; \theta)$ to any specific parametrization, and we carry the explicit dependence of $x_s$ on the variational parameters, $x_s \rightarrow x_s (\theta)$.

We can define a cost function, related to the LER, that can be optimized with respect to the noise model.
A common choice is the binary cross entropy
\begin{equation}\label{eq:bin_cross_ent}
    \begin{aligned}
    L(\theta)=-\sum_{n=1}^{|\mathcal D|}
    \big[&y_n\log x_{s_n}(\theta)\\
    &+(1-y_n)\log\!\left(1-x_{s_n}(\theta)\right)\big].
    \end{aligned}
\end{equation}

To analyze this loss, it is useful to make explicit the data-generating process behind $\mathcal{D}$. We assume that the dataset is composed of i.i.d.\ samples $(s_n,y_n)\sim\mathcal P^\star$ from the joint distribution induced by the device's true (unknown) noise model $P_N^\star$:
\begin{equation}\label{eq:true_data_dist}
    \mathcal{P}^\star(s,y) \;=\; \sum_e P_H(e,s)\,P_L(e,y)\,P_N^\star(e),
\end{equation}
i.e.\ the joint distribution of Eq.~\eqref{eq:full_distribution} marginalized over the (latent) error string $e$, evaluated at $P_N = P_N^\star$. The empirical loss in Eq.~\eqref{eq:bin_cross_ent} is then a Monte-Carlo estimator of the corresponding population loss, $L(\theta) \approx |\mathcal{D}| \cdot \mathcal{L}[x_{\cdot}(\theta)]$, where the dot is a placeholder for the syndrome index, so $x_{\cdot}(\theta)$ denotes the full predictor $s\mapsto x_s(\theta)$, and $\mathcal{L}$ is defined below in Eq.~\eqref{eq:population_loss}.

The following proposition makes precise the sense in which minimizing $L(\theta)$ is the ``right'' thing to do: for a fixed code and logical observable, the minimizer of the population cross-entropy induces the decoder that attains the information-theoretically lowest possible logical error rate.

\begin{proposition}[Optimality of the variational principle]\label{prop:optimality}
Fix the parity-check matrix $H$ and the logical parity-check rows $L$, and let the data be drawn from $\mathcal{P}^\star$ as in Eq.~\eqref{eq:true_data_dist}. Define the true posterior
\begin{equation}\label{eq:true_posterior}
    \eta(s) \;:=\; \frac{p^\star(s, 1)}{p^\star(s, 0) + p^\star(s, 1)},
\end{equation}
where $p^\star(s, \ell)$ is the marginal of Eq.~\eqref{eq:l_mariginal} evaluated with $P_N = P_N^\star$. For any measurable predictor $x: s \mapsto [0, 1]$, let
\begin{equation}\label{eq:population_loss}
    \mathcal{L}[x] = -\mathbb{E}_{(s,y) \sim \mathcal{P}^\star}\!\left[y\log x(s) + (1-y)\log(1-x(s))\right].
\end{equation}
Then:
\begin{enumerate}
    \item[(i)] $\mathcal{L}[x]$ is uniquely minimized, pointwise in $s$, by $x^\star(s) = \eta(s)$.
    \item[(ii)] Among all syndrome-based decoders, the Bayes rule $\widehat y^\star(s)=\mathbf{1}[\eta(s)>1/2]$ attains the information-theoretic minimum expected logical error rate
    \begin{equation}\label{eq:bayes_ler}
        \mathrm{LER}^\star = \mathbb{E}_s\!\left[ \min\{\eta(s), 1 - \eta(s)\} \right].
    \end{equation}
    \item[(iii)] If the ansatz $\{P_N(e; \theta)\}_{\theta \in \Theta}$ is expressive enough that there exists $\theta^\star$ with $x_s(\theta^\star) = \eta(s)$ for all $s$, then $\theta^\star$ is a global minimizer of $\mathcal{L}$ and the induced decoder attains $\mathrm{LER}^\star$.
\end{enumerate}
\end{proposition}

\noindent\textit{Proof sketch.} For any fixed $s$, the integrand in Eq.~\eqref{eq:population_loss} reduces to $-q \log x - (1-q)\log(1-x)$ with $q = \mathbb{E}[y \mid s] = \eta(s)$, which is strictly convex in $x \in (0,1)$ and uniquely minimized at $x=q$. This proves (i). Part (ii) is the classical Bayes-rule bound: Eq.~\eqref{eq:ler} is the empirical $0$--$1$ risk of the hard classifier $\widehat y$, its expectation under $\mathcal P^\star$ is the population risk $\mathbb E_s[\,\mathbb P(\widehat y(s)\neq y\mid s)\,]$, and thresholding $\eta(s)$ at $1/2$ minimizes the latter pointwise, giving Eq.~\eqref{eq:bayes_ler}. Part (iii) follows by composing (i) and (ii). \hfill$\square$

\textit{Remark (identifiability up to observational equivalence).} The map $P_N^\star \mapsto\mathcal P^\star$ is many-to-one. For $\phi:e\mapsto(He,Le)$, the silent fault combinations form
\begin{equation}\label{eq:silent_kernel}
    K \;=\; \{e \in \mathbb{F}_2^N \,:\, He = 0 \text{ and } Le = 0\},
\end{equation}
whose elements produce no syndrome or selected logical flip, that is, stabilizers of the code.
Two noise models yield the same joint distribution precisely when their total probability agrees on every coset of $K$.
This non-identifiability is well documented in the syndrome-based noise-learning literature~\cite{wagner2021optimal, wagner2022pauli, wagner2023learning, zheng2025efficient, iyer2025enhancing}.
Crucially, the loss identifies only $\eta(s)=\mathcal P^\star(\ell=1\mid s)$, not $\mathcal P^\star(s)$. Thus Proposition~\ref{prop:optimality} establishes optimal decoding, identifying $\theta^\star$ only up to this equivalence class.

\textit{Remark (choice of loss).} Because the binary cross-entropy is a strictly proper scoring rule, its population minimizer coincides with that of any other strictly proper scoring rule, such as the squared loss $\sum_s (x_s - y_s)^2$; switching between them does not change the optimum.
The binary cross-entropy is strictly convex as a function of the predicted probability $x$, but the optimization variable here is $\theta$; because $x_s(\theta)$ is a nonlinear rational function induced by decoder marginalization, $L(\theta)$ is generally non-convex.

In the remainder of this section we derive the gradient of $L(\theta)$, so that the variational problem of Proposition~\ref{prop:optimality} can be attacked by gradient descent whenever $P_N(e; \theta)$ is differentiable in $\theta$.
When the decoder itself is also differentiable, the gradient can be evaluated exactly; otherwise, finite-difference or derivative-free methods such as Nelder-Mead remain viable.

We stress that the optimal noise model that results from this procedure does not need to provide a faithful representation of all error mechanisms in the device.
It is completely biased towards the specific code, the logical observable, and the choice of the decoder.
When the primary goal is quantum error correction rather than device tomography, this bias focuses optimization on task-relevant noise structure.
The learned parameters are task-, ansatz-, and decoder-dependent, so physical interpretation requires independent validation.

Analytical gradients are important in practice because they can substantially reduce optimization cost compared with finite-difference or derivative-free alternatives.
Let us assume a parametrization of the noise model $P_N (e; \theta)$ that depends on a set of learnable variational parameters $\theta = \{\theta_1, \theta_2, \ldots, \theta_N\}$.
The gradient of $L(\theta)$ in Eq.~\eqref{eq:bin_cross_ent} can be expressed as
\begin{equation}\label{eq:grad_cross_ent}
    \partial_{\theta_i} L(\theta) = - \sum_{s \in \mathcal{D}} \frac{y_s - x_s (\theta)}{x_s (\theta) (1- x_s (\theta))} \partial_{\theta_i} x_s (\theta),
\end{equation}
where, from Eq.~\eqref{eq:log_flip_prob},
\begin{equation}\label{eq:grad_prediction}
    \partial_{\theta_i} x_s (\theta) = \frac{p(s, 0) \partial_{\theta_i} p(s, 1) - p(s, 1) \partial_{\theta_i} p(s, 0)}{(p(s, 1) + p(s, 0))^2},
\end{equation}
and
\begin{equation}\label{eq:grad_partition}
    \partial_{\theta_i} p(s, \ell, \theta) = \sum_e P_H (e, s)P_L (e, \ell) \partial_{\theta_i} P_N (e, \theta)
\end{equation}
from Eq.~\eqref{eq:full_distribution} and \eqref{eq:l_mariginal}.
In Eq.~\eqref{eq:grad_partition}, we made the dependence of $p(s, \ell)$ on the noise model's variational parameters explicit.
Therefore, to express $\partial_{\theta_i} L(\theta)$ analytically, we just need to know $\partial_{\theta_i} P_N (e, \theta)$.
Given an explicit form of $\partial_{\theta_i} P_N(e, \theta)$, the gradients of $L(\theta)$ follow directly from the chain rule above.

The objective is not intrinsically tied to tensor-network decoding. In principle, an approximate decoder could support noise learning if it provides noise-dependent estimates of the logical-class probabilities $p(s,\ell;\theta)$ together with derivatives or a separate gradient estimator. For a differentiable approximate decoder, differentiation yields gradients of its approximate probabilities; although these need not coincide with the exact-model gradients, they may still provide a useful optimization signal. A sampling decoder could instead estimate logical-class probabilities by assigning each sampled error string to a logical class and averaging the resulting outcomes. Hard-decision output alone is insufficient without such an augmentation. These possibilities suggest a route to trading accuracy for computational cost, but we do not test them here.

As a concrete example, consider the uncorrelated noise model $P_N(e; \theta) = \prod_i P_i(e_i; \theta_i) = \prod_i \theta_i^{e_i} (1 - \theta_i)^{1 - e_i}$.
Since the factors are independent, the derivative with respect to a single parameter $\theta_j$ takes the simple form $\partial_{\theta_j} P_N(e; \theta) = (2e_j - 1) \prod_{i \neq j} P_N(e_i; \theta)$,
which, substituted into Eq.~\eqref{eq:grad_partition}, yields
\begin{equation}\label{eq:grad_uncorrelated}
    \partial_{\theta_j} p(s, \ell; \theta) = p(s, \ell; \theta, \theta_j=1) - p(s, \ell; \theta, \theta_j=0).
\end{equation}
Equation~\eqref{eq:grad_uncorrelated} reduces each derivative required for noise learning to a difference between two logical-class probabilities evaluated with the $j$th noise parameter clamped to zero and one. Many scalable approximate decoders, including matching-based and belief-propagation methods, already consume noise-dependent weights. Standard hard-decision implementations do not directly provide the clamped logical-class probabilities on the right-hand side of Eq.~\eqref{eq:grad_uncorrelated}; however, augmenting such decoders with probabilistic or sampling-based estimators could provide approximate gradients. Investigating this approach is an interesting direction for future work and could enable noise learning at code distances beyond the reach of exact contraction. This extension is not demonstrated here.

\section{Tensor-network decoder}\label{sec:TN_decoder}

\begin{figure}[t!]
    \centering
    \includegraphics[width=1.0\linewidth]{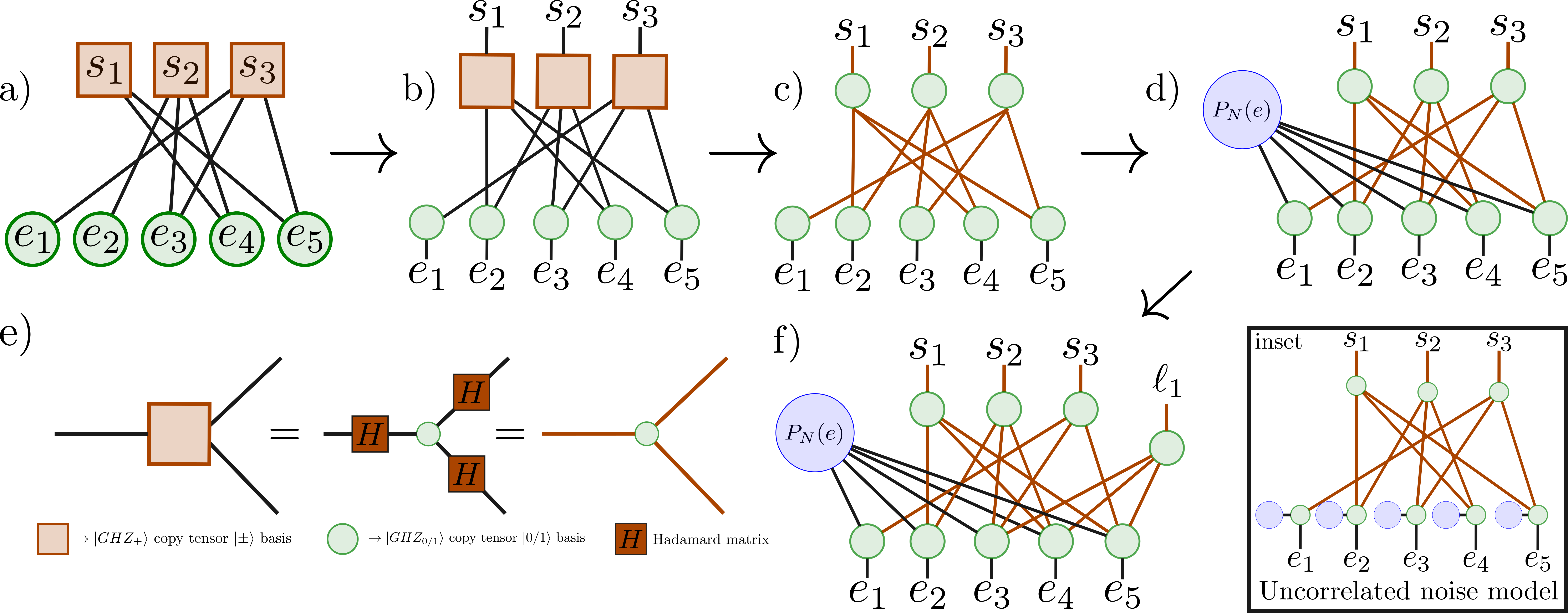}
    \caption{Mapping from the Tanner graph to a tensor-network decoder.
    \textbf{a)}~The Tanner graph with syndrome nodes $s_1, s_2, s_3$ (orange squares) and error nodes $e_1, \ldots, e_5$ (green circles).
    \textbf{b)}~Each error node becomes a $|GHZ_{0/1}\rangle$ copy tensor (green) and each syndrome node a $|GHZ_{\pm}\rangle$ copy tensor (orange).
    \textbf{c)}~Inserting Hadamard matrices $H$ on the bonds between syndrome and error tensors converts all copy tensors to the same ($|0/1\rangle$) basis, as explained in~\textbf{e)}.
    \textbf{d)}~A noise model $P_N(e)$ is attached to the error legs, adding a new tensor that encodes the fault-event probabilities.
    \textbf{e)}~Identity: a $|GHZ_{\pm}\rangle$ copy tensor equals a $|GHZ_{0/1}\rangle$ copy tensor dressed with Hadamard matrices on each leg.
    \textbf{f)}~Full decoder network including a logical observable $\ell_1$ as an additional check node. This panel also illustrates the special case of an uncorrelated noise model, where $P_N(e)$ factorizes into independent single-variable tensors.
    All tensor-network objects in this figure, copy tensors, Hadamard gates, noise tensors, and their contractions, can be straightforwardly constructed and manipulated using the \texttt{quimb} library~\cite{gray2018quimb} and NVIDIA's CUDA-Q package~\cite{kim2023cudaq}.}
    \label{fig:tn_mapping}
\end{figure}

In the remainder of this paper, we specialize to tensor-network (TN) decoders, which allow us to compute $\partial_{\theta_i} L(\theta)$ exactly.
Exact contractions are generally tractable only for small-to-moderate problem sizes. However, exactness provides high-quality gradient signals and a useful reference for approximate methods; this is the key trade-off between accuracy and computational cost.
To map the problem from the previous sections onto a tensor network, note that both $P_H$ and $P_L$ are expressed as a product of \textit{local} terms in Eq.~\eqref{eq:P_H} and Eq.~\eqref{eq:P_L}, respectively, where ``local'' means each factor depends on a small subset of variables determined by the Tanner-graph connectivity.
Functions that can be expressed as products of local terms are sometimes called factor graphs or graphical models.
Thus, $P(e, s, \ell)$ in Eq.~\eqref{eq:full_distribution} is a factor graph as long as the noise model $P_N$ can be factorized.
These include popular noise models such as uncorrelated, Gaussian, locally correlated, and tensor-network models.
A factor graph supported on a discrete domain (in what follows, binary values) admits a straightforward mapping to a tensor network.
The mapping works in two steps.
First, for each error variable $e_i$ we define a copy tensor
\begin{equation}
\begin{split}
    t_{ijk} &= 1, \text{ if } i = j = k \\
    &= 0, \text{ otherwise}
\end{split}
\end{equation}
with as many indices as the degree of the vertex corresponding to $e_i$ in the Tanner graph.
The indices carry discrete states, while the tensor entries encode compatibility weights (and, once multiplied by noise factors, probabilities).
Second, we map each factor $I_a (e_i^{s_a}, e_j^{s_a}, \ldots, e_k^{s_a}, s_a)$ in Eq.~\eqref{eq:P_H} to a tensor.
Since $I_a (\cdot)$ is supported on a discrete domain, it can already be seen as a tensor with indices $\{ e_i^{s_a}, e_j^{s_a}, \ldots, e_k^{s_a}, s_a \}$ and entries defined as in Eq.~\eqref{eq:indicator_function}.
In our particular scenario, $I_a(\cdot)$ computes the parity of its argument.
It is easy to show that it corresponds to a copy $| GHZ_{\pm} \rangle \propto |+\rangle^{\otimes 4} + |-\rangle^{\otimes 4}$ tensor in the $| \pm \rangle \propto |0\rangle \pm |1\rangle$ basis
\begin{equation}
    I(e_i, e_j, e_k, s_a) \propto \langle e_i, e_j, e_k, s_a |GHZ_{\pm} \rangle \propto \langle e_i, e_j, e_k, s_a | H^{\otimes 4} | GHZ_{0/1} \rangle,
\end{equation}
where in the second equality we used the Hadamard matrix $H$ and the copy $| GHZ_{0/1} \rangle \propto |0\rangle^{\otimes 4} + |1\rangle^{\otimes 4}$ tensor in the $| 0/1 \rangle$ basis.
In other words, the tensor network that represents $P_H (e, s)$ is made up of two sets of copy tensors.
One copy tensor on the basis $| 0/1 \rangle$ for each error variable and one copy tensor on the basis $| \pm \rangle$ for each syndrome variable.

In Fig.~\ref{fig:tn_mapping}a) we reproduce the Tanner graph from the previous section.
Fig.~\ref{fig:tn_mapping}b) shows how to map it onto a tensor network with copy tensors in the $| 0/1 \rangle$ (green circles) and $| \pm \rangle$ (orange squares).
In Fig.~\ref{fig:tn_mapping}c) we use Hadamard matrices to simplify the tensor network.
By inserting a Hadamard matrix at each bond (see orange edges and Fig.~\ref{fig:tn_mapping}e)), we reduce the problem to just one type of copy tensor.
Notice that this is computationally desirable, since copy tensors can be stored lazily as indices without allocating any memory.

We now attach the noise model.
In Fig.~\ref{fig:tn_mapping}d) we connect a generic noise model to the original tensor network by adding an index to the copy tensors corresponding to the error variables.
Any classical multivariate distribution over error bit-strings that admits a compact tensor network representation can be employed.
These include uncorrelated noise models as in the previous section (see panel \textbf{d} in Fig.~\ref{fig:tn_mapping}), Matrix Product States (MPS)~\cite{white1992density, schollwock2011density}, Projected Entangled Pair States (PEPS)~\cite{verstraete2004renormalization}, and more general tensor networks~\cite{gray2021hyper, gray2024hyperoptimized}.
Stronger and non-local correlations in the noise model will result in higher contraction costs.

As we commented on in Eq.~\eqref{eq:code_noise}, in order to implement an \textit{exact} most-likely correction decoder, we can fix the dangling $s_i$ legs of the tensor network in Fig.~\ref{fig:tn_mapping}d) with a given syndrome and contract or sample the compatible error configurations and aggregate their probabilities by correction class.
The most-likely correction decoder will be exact when the class probabilities are contracted exactly.
Selecting the single most-likely error is generally only an approximation to selecting the most-likely correction, because many individually less likely errors can belong to the same correction class and collectively carry greater probability. Tensor networks are well suited to the latter task because the class constraint can be represented by an additional parity or correction leg and all compatible errors can be summed in one contraction.

\begin{figure}[t!]
    \centering
    \includegraphics[width=0.5\linewidth]{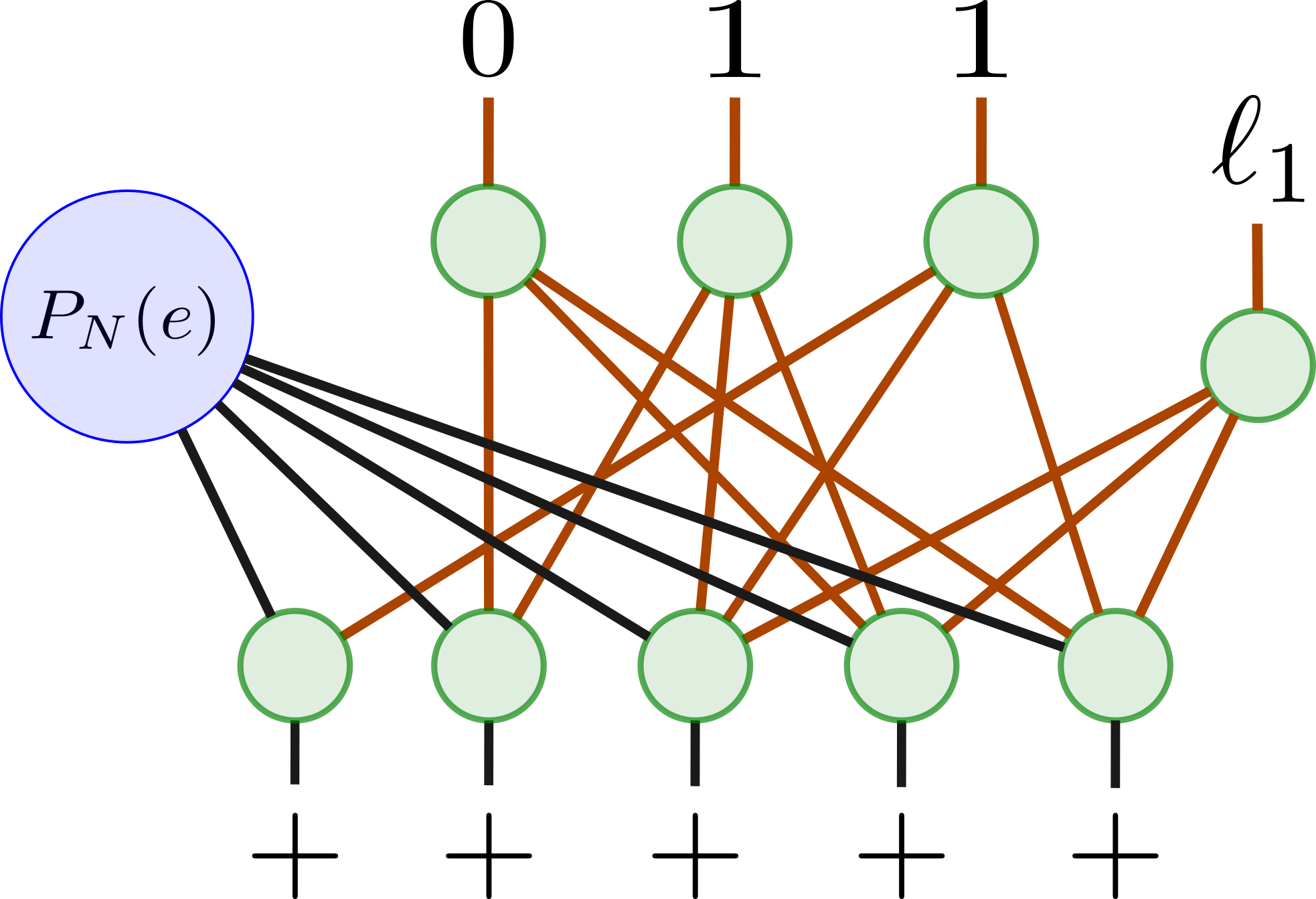}
    \caption{Maximum-likelihood decoding as a tensor-network contraction. The syndrome legs are fixed to $(s_1, s_2, s_3) = (0,1,1)$, each error variable is contracted with a $|+\rangle$ state (crosses at the bottom) to sum over all error configurations, and the noise model $P_N(e)$ is attached on the left. The logical observable $\ell_1$ remains as a free dangling leg (top right), so that the contraction yields the two-component vector $p(s, \ell_1)$.}
    \label{fig:mld}
\end{figure}

On the other hand, in order to construct an exact Maximum-Likelihood decoder, we need at least one logical observable.
As we discussed in Eq.~\eqref{eq:P_L}, any logical observable can be modeled as an additional check variable in the parity-check matrix.
It is therefore straightforward to add an observable $\ell_1$ to the network, as shown in Fig.~\ref{fig:tn_mapping}f).
Now we have all the ingredients to compute the likelihood that a given syndrome flips an observable.
Since we want to sum over all possible error configurations, we contract each error leg with a $|+\rangle$ state.
The dangling $s_i$ legs are instead fixed to a given syndrome configuration, for example $(s_1, s_2, s_3) = (0,1,1)$.
This procedure is illustrated in Fig.~\ref{fig:mld}, which shows the full tensor network ready for contraction: the syndrome legs are clamped to $(0,1,1)$, each error variable is projected onto $|+\rangle$ to perform the sum over all error configurations, and the logical observable $\ell_1$ remains as a free leg.
The contraction will result in a two-dimensional vector that corresponds exactly to $p(s, \ell_1)$ in Eq.~\eqref{eq:l_mariginal} with $s = (0,1,1)$.
From $p(s, \ell_1)$ we can finally compute $x_s$ in Eq.~\eqref{eq:log_flip_prob}.

\subsection{Learning noise models without prior knowledge}

We now apply the noise-learning framework of Sec.~\ref{sec:NML} using the tensor-network decoder described above.
Since both the likelihood $x_s$ and the noise model $P_N(e;\theta)$ admit tensor-network representations, the gradient $\partial_{\theta} L(\theta)$ can be computed analytically.

We initialize all parameters to the same value, without mechanism-specific prior information about the error rates.
The only input to the optimization is a dataset of syndromes and logical observable flips, which can be collected from a standard memory experiment, or synthetically generated.
From this initialization, the optimization finds a model that matches models obtained through detailed device characterization.
This setting is particularly relevant for scenarios where explicit noise tomography is impractical, such as when the device is not fully accessible or when the noise landscape changes between calibration cycles.
For a different fault-tolerant gadget, the problem must be formulated using the appropriate figure of merit. For example, a logical Clifford gate can be implemented as a Clifford circuit, decoded from its detector outcomes, and optimized using the logical gate fidelity as a proxy.

To demonstrate how this works for a memory experiment, we study a distance-3 surface code with 3 rounds of syndrome extraction.
The syndrome data is sampled from the detector error models (DEMs) extracted from Google's Sycamore processor~\cite{acharya2023suppressing}, which provides an experimentally realistic benchmark for the noise-learning procedure.
Thus, the benchmark uses synthetic syndrome and logical-observable samples generated from an experimentally characterized hardware DEM, rather than shots collected directly from a live device.
We use an uncorrelated noise model ansatz $P_N(e; \theta) = \prod_i P_i(e_i; \theta_i)$ and optimize the binary cross entropy in Eq.~\eqref{eq:bin_cross_ent} via stochastic gradient descent.
At each of $500$ optimization iterations, the gradients are estimated using a batch of $30{,}000$ syndrome shots.
The MSE is averaged over fault-event probabilities. LER curves show the mean and one standard deviation over 50 fresh evaluation batches.
As a baseline, we compare our learned noise model against the DEM provided by Google~\cite{acharya2023suppressing_data}, which was obtained through an independent and detailed characterization of the device.

\begin{figure}[t!]
    \centering
    \includegraphics[width=\textwidth]{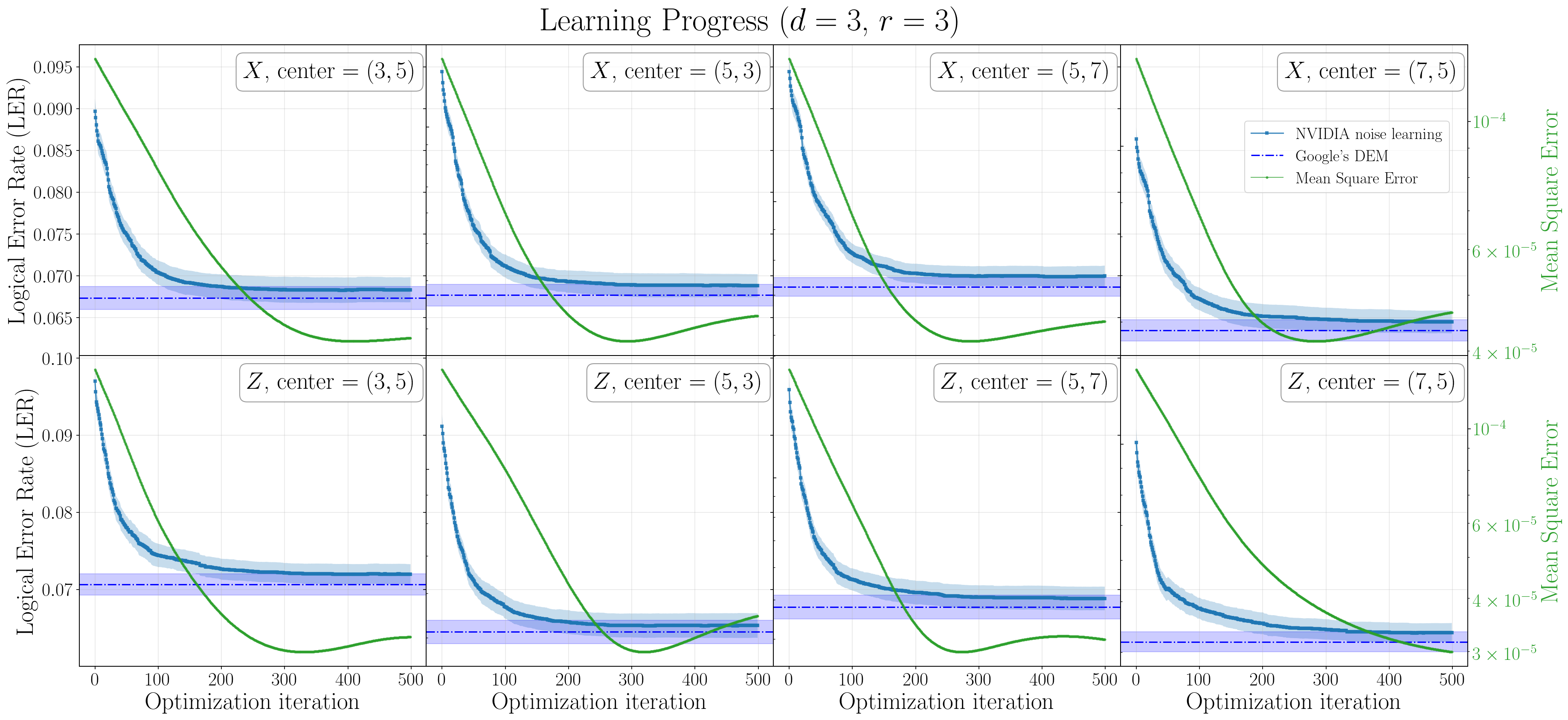}
    \caption{Learning progress for a distance-3 surface code with 3 rounds of syndrome extraction. Top row: X logical observable. Bottom row: Z logical observable. Each panel corresponds to a different center qubit position on the device. The logical error rate (LER, left axis) decreases as a function of the optimization iteration, converging towards the baseline set by Google's detector error model (DEM). The mean squared error (MSE, right axis, green) between the learned noise parameters $\theta$ and the DEM parameters is shown on a logarithmic scale. The MSE follows a clear downward trend throughout the training, indicating that the optimization not only improves decoding performance but also recovers noise parameters that are close to those obtained from explicit device characterization.}
    \label{fig:ler_comparison}
\end{figure}

The learning progress is shown in Fig.~\ref{fig:ler_comparison} for both the X and Z logical observables across four different regions of the device.
Starting from a uniform initialization of the noise model parameters, the logical error rate (LER) drops rapidly during the first $\sim 100$ iterations and then gradually converges towards the baseline set by Google's DEM.
After $500$ iterations, the learned noise model reaches a LER that is within the statistical uncertainty of the reference value across all configurations, using only synthetic syndrome and logical-observable samples generated from the experimentally characterized hardware DEM, without direct access to the hardware or additional explicit noise tomography.
The shaded regions in the figure reflect the variance across different syndrome batches and confirm that the optimization remains stable throughout the training.
In addition to tracking the LER, we monitor the mean squared error (MSE) between the learned parameters $\theta$ and the reference DEM parameters (right axis, green curves in Fig.~\ref{fig:ler_comparison}).
The MSE shows a clear downward trend on a logarithmic scale, supporting (though not by itself proving) that the optimization is not merely finding parameters that happen to decode well, but is genuinely recovering the underlying noise structure of the device.
This provides additional evidence that the learned noise model captures physically meaningful information about the error processes, beyond what is needed for optimal decoding alone.
These results test decoding on memory experiments drawn from the same family of characterized DEMs used for training. They therefore do not by themselves demonstrate generalization to more complex logical operations. Generalization will depend on whether the learned parameters represent physical mechanisms shared by the memory and operation circuits; operation-specific mechanisms would require additional labeled data and model parameters.
We report distance-3 results because the cost of the exact tensor-network contractions used to obtain reference-quality gradients grows exponentially with graph treewidth, while the number and spacetime connectivity of detector events also increase with code distance and rounds. A distance-5 memory circuit is therefore substantially more expensive than the distance-3 benchmark and was outside the exact-contraction study presented here. The variational objective itself is not restricted to distance 3: approximate tensor-network contraction or a scalable decoder such as matching, Union-Find, belief propagation, or a suitably conditioned neural network can provide the forward estimates needed at higher distance, as discussed in Sec.~\ref{sec:NML} and below.

\subsection{Online learning of device drifts}

\begin{figure}[t!]
    \centering
    \includegraphics[width=\textwidth]{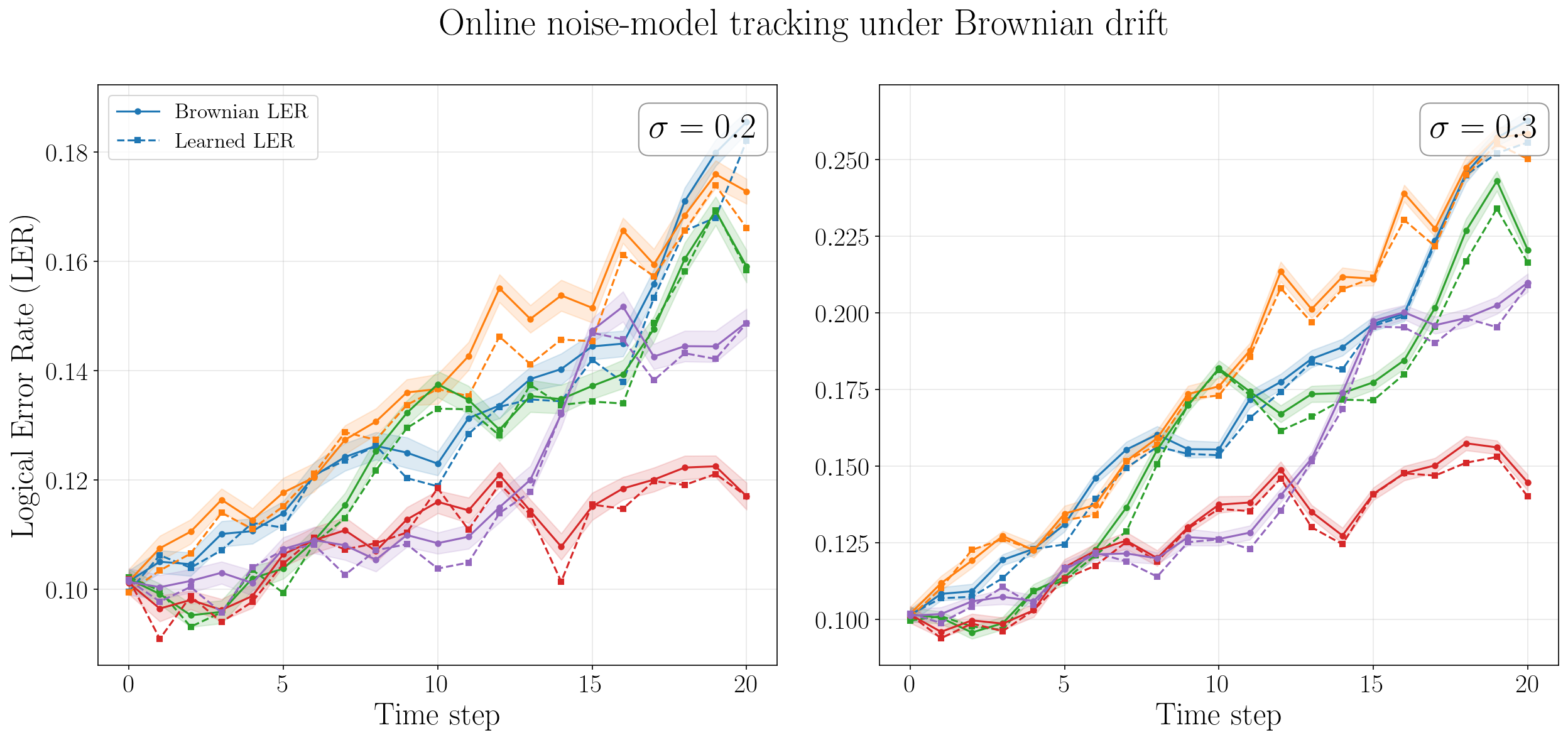}
    \caption{Online noise-model tracking under synthetic device drift for a distance-3 surface code with 3 rounds. The noise model evolves over 20 time steps along a Brownian motion in logit space (Eq.~\eqref{eq:brownian_logit}) with drift amplitude $\sigma = 0.2$ (left) and $\sigma = 0.3$ (right). Five independent random seeds are shown, each as a different color. For each seed, solid lines with circles (``Brownian LER'') show the LER of a decoder with perfect knowledge of the true drifted noise model, while dashed lines with squares (``Learned LER'') show the LER of the online-learned model, which is warm-started from the previous time step. The learned model tracks the true baseline closely across all seeds, even under the larger drift of $\sigma = 0.3$.}
    \label{fig:brownian_tracking}
\end{figure}

In practice, the noise characteristics of a quantum device are not static.
Error rates drift over time due to fluctuations in control electronics, defects, imperfections, and other environmental factors.
A noise model that was accurate at calibration time can become stale within minutes to hours, degrading decoder performance.
This motivates an \textit{online} learning strategy in which the noise model is continuously updated as fresh syndrome data becomes available, without the need to restart the optimization from scratch.
A related approach using reinforcement learning to adapt quantum control parameters in real time has recently been demonstrated by Sivak et al.~\cite{sivak2025reinforcement}, where error detection events serve as learning signals to stabilize surface code performance under drifting noise.

To study this scenario in a controlled setting, we synthetically model device drift by evolving the error probabilities of the detector error model along a Brownian motion in logit space.
Let $\theta^{(0)} = \{\theta_1^{(0)}, \ldots, \theta_N^{(0)}\}$ denote the initial error probabilities extracted from the device characterization.
We define a time-dependent noise model $\theta^{(t)}$ through the following stochastic process.
First, we map each probability to its logit,
\begin{equation}\label{eq:logit}
    \lambda_i^{(0)} = \log \frac{\theta_i^{(0)}}{1 - \theta_i^{(0)}}.
\end{equation}
The logit-space parameters then evolve according to a discrete Brownian motion,
\begin{equation}\label{eq:brownian_logit}
    \lambda_i^{(t)} = \lambda_i^{(0)} + \sigma \sum_{\tau=1}^{t} \xi_i^{(\tau)}, \qquad \xi_i^{(\tau)} \sim \mathcal{N}(0, 1),
\end{equation}
where $\sigma$ controls the amplitude of the drift and $\xi_i^{(\tau)}$ are independent standard normal increments.
Finally, the physical error probabilities at time $t$ are recovered via the sigmoid function,
\begin{equation}\label{eq:sigmoid_recovery}
    \theta_i^{(t)} = \frac{1}{1 + e^{-\lambda_i^{(t)}}}.
\end{equation}
Operating in logit space ensures that the error probabilities remain in $(0,1)$ throughout the evolution.
As the drift amplitude $\sigma$ increases, the noise parameters explore a wider region of the hypercube $[0,1]^N$, modeling more severe device instabilities.

At each time step $t$, new datasets $\mathcal{D}^{(t)}$ of syndromes and logical observable flips are sampled from the new detector error model at time $t$, simulating the performance degradation of real devices.
Rather than reinitializing the variational parameters, the optimization continues from the parameters learned at the previous time step.
This warm-starting strategy exploits the temporal continuity of the drift: since the noise model changes only incrementally between consecutive steps, the previously learned parameters provide a good initial condition, and only a small number of gradient updates are needed to track the new noise landscape.

In Fig.~\ref{fig:brownian_tracking} we demonstrate this online learning procedure for a memory experiment on a distance-3 surface code with 3 rounds.
The detector error model is evolved over $20$ time steps along a Brownian path with two drift amplitudes, $\sigma = 0.2$ (left) and $\sigma = 0.3$ (right).
To assess robustness, the experiment is repeated across five independent random seeds, each shown as a separate color.
At each step, the decoder is re-optimized using a fresh batch of syndrome data sampled from the current detector error model.
Solid lines with circles show the logical error rate obtained by a decoder that has perfect knowledge of the true (drifted) noise model at each time step, serving as the best achievable baseline.
Dashed lines with squares show the LER of the online-learned model.
Across all seeds and both drift-amplitude settings, the learned noise model tracks the drifting baseline closely, with the two curves remaining in close agreement at nearly every time step.
As expected, larger drift amplitude ($\sigma = 0.3$) leads to more pronounced fluctuations in the LER, yet the online-learned model continues to follow the baseline faithfully.
This result confirms that the warm-starting strategy is effective: the optimization requires only a modest number of gradient updates per step to adapt to the incremental changes in the noise landscape.

The ability to track device drifts online has important practical implications.
It effectively recalibrates the noise model, which would otherwise require dedicated characterization experiments.
Instead, the decoder continuously refines its noise model using the same syndrome data that is already collected during normal operation, making the approach compatible with real-time decoding and calibration pipelines. The learned drift can flag when discrete recalibration is needed or support continuous calibration while the device operates.

\section{Conclusion}

We have presented a variational framework for learning the noise model of a quantum error-correcting code directly from syndrome and logical observable data collected during hardware memory experiments.
The method formulates noise characterization as a variational optimization problem: the fault-event probabilities are treated as variational parameters and updated to minimize the binary cross-entropy between predicted and observed logical flips.
The objective is principled rather than ad hoc: we prove (Proposition~\ref{prop:optimality}) that its population minimizer coincides with the Bayes-optimal posterior decoder, so that any sufficiently expressive ansatz attains the information-theoretic minimum logical error rate, up to the well-known stabilizer-coset equivalence class.
We derived the required gradients analytically within the tensor-network decoding formalism and demonstrated the approach on circuit-level noise data extracted from Google's Sycamore processor.
The main results are:
\begin{itemize}
    \item The variational objective is provably the right thing to optimize: its population minimizer is the Bayes-optimal decoder and saturates the information-theoretic LER lower bound (Proposition~\ref{prop:optimality}, Sec.~\ref{sec:NML}).
    \item Starting from a completely uninformed initialization, the optimization recovers a noise model whose logical error rate matches that of Google's independently characterized detector error model, using only syndrome data as input (Sec.~\ref{sec:TN_decoder}). Beyond matching LERs, the mean square error between learned and reference parameters decreases throughout training, indicating that the recovered noise structure is physically meaningful and not merely an LER-equivalent surrogate.
    \item Under synthetic device drift modeled by Brownian motion in logit space, the learned noise model tracks the evolving noise landscape in real time via warm-started optimization, maintaining decoding performance comparable to a decoder with perfect knowledge of the instantaneous noise.
\end{itemize}

The method has several notable strengths.
Because the tensor-network decoder is exact (or systematically improvable), the resulting noise estimates inherit this precision: the optimizer receives high-quality gradient signals that reflect the true sensitivity of the logical error rate to each noise parameter.
The framework naturally accommodates correlated noise (Fig.~\ref{fig:tn_mapping}d), since the noise model ansatz can be extended beyond independent error channels to include multi-qubit correlations without modifying the optimization procedure.
The learned output is an explicit, interpretable noise model in the form of a tensor network that can be inspected, transferred to other decoders, or used for diagnostic purposes.
The optimization can also use a neural-network decoder as its computational backbone, provided that the network outputs probabilities conditioned on the noise parameters or can be re-evaluated as those parameters change. A conventional end-to-end neural decoder that does not expose any dependence on $P_N(e;\theta)$ cannot by itself supply gradients for learning $\theta$; it must be augmented, conditioned, or retrained within the optimization loop.
The online learning capability further allows the model to adapt continuously to device drifts without interrupting the computation.
Finally, the mathematical transparency of the tensor-network formulation ensures that every step --- from syndrome processing to gradient computation --- is analytically understood and free of black-box components.

From a practical standpoint, the approach can substantially simplify noise characterization in experimental settings.
The only data required are syndrome measurements and logical observable outcomes, both of which are already collected during standard QEC memory experiments.
No dedicated offline characterization protocols are needed: the decoder learns the noise model in situ, as a byproduct of its normal operation.

The main open questions concern scalability and extension beyond memory-style logical-observable prediction.
Exact tensor-network contraction scales exponentially with the treewidth of the underlying factor graph, which grows with code distance and the number of syndrome extraction rounds.
For moderate code sizes considered here, the contraction is tractable, but extending to distance-5 and higher-distance codes will require approximate contraction strategies.
Three complementary directions are promising.
First, compositional approaches can decompose the full decoding problem into smaller, independently contractible sub-instances --- for example, by partitioning the spacetime decoding graph along spatial or temporal boundaries and combining the results.
Second, more advanced tensor-network contraction techniques, such as belief propagation~\cite{alkabetz2021tensor, tindall2023gauging, pancotti2023one, evenbly2024loop} on the factor graph or approximate message-passing schemes, can replace exact contraction while preserving differentiability and thus the ability to compute gradients for noise learning.
Third, as emphasized in Sec.~\ref{sec:NML}, the analytical gradient expressions in Eqs.~\eqref{eq:grad_cross_ent}--\eqref{eq:grad_partition} are decoder-agnostic: any decoder that produces an estimate of the logical-observable probability $x_s(\theta)$, or the analogous probability of a correction class, can be used to approximate the gradients, even if the decoder itself is not a tensor network.
This opens the possibility of using more scalable decoders --- such as Minimum Weight Perfect Matching, Union-Find, or neural-network decoders --- as the forward pass for gradient estimation, while retaining the noise model parameterization and optimization framework developed here.
The resulting gradients would be approximate, but should still be sufficient to drive the noise model towards improved decoding performance at code distances well beyond the reach of exact tensor-network contraction.
Combining these strategies with hardware-aware parallelism, and integrating the optimizer directly into real-time decoding pipelines, could extend the reach of the method to the code distances required for practical fault-tolerant quantum computation.

\section*{Acknowledgments}
We thank Taylor Patti, Ben Howe, Justin Lietz, and Fernando Pastawski for helpful discussions and comments on the manuscript.

\bibliographystyle{plain}
\bibliography{bibliography}

\end{document}